\documentclass[11pt,times]{article}
\usepackage{graphicx}
\usepackage[T1]{fontenc} 

\usepackage{epsfig}
\usepackage[english]{babel}
\usepackage{amsfonts,amsmath,enumerate}
\usepackage{amscd}
\usepackage{amsthm}
\usepackage{amsmath,amsfonts,amssymb,mathrsfs}
\usepackage{latexsym}
\usepackage{amssymb}
\usepackage{color}
\usepackage{xifthen}

\usepackage[numbers]{natbib}%
\usepackage[titletoc]{appendix}

\newlength{\oldparindent}
\usepackage{fancyhdr}
\usepackage{hyperref} 
\hypersetup{
	colorlinks = true,
	linkcolor = blue,
	anchorcolor = blue,
	citecolor = blue,
	filecolor = blue,
	urlcolor = blue
}

\usepackage{makeidx}
\usepackage{hyperref}
\usepackage{amsmath}
\usepackage{amsthm}
\usepackage{amsfonts}
\usepackage{amssymb}
\usepackage{mathrsfs}
\usepackage{graphicx}
\usepackage{tikz,tikz-cd}
\usetikzlibrary{matrix}
\usepackage{smartdiagram}
\makeatletter
\newcommand{\superimpose}[2]{%
	{\ooalign{$#1\@firstoftwo#2$\cr\hfil$#1\@secondoftwo#2$\hfil\cr}}}
\makeatother

\newcommand{\rr}{{\mathbb{R}}}
\newcommand{\rrD}{{\rr^D}}

\newcommand{\ee}{{\mathbb{E}}}

\newcommand{\pp}{{\mathbb{P}}}
\newcommand{\qq}{{\mathbb{Q}}}

\newcommand{\fff}{{\mathscr{F}}}

\newcommand{\ffft}{{\mathscr{F}_t}}

\makeatletter
\def\@chapter[#1]#2{\ifnum \c@secnumdepth >\m@ne
	\refstepcounter{chapter}%
	\typeout{\@chapapp\space\thechapter.}%
	\addcontentsline{toc}{chapter}%
	{\protect\numberline{\thechapter}\string\hypertarget{chap\thechapter}{#1}}%
	\else
	\addcontentsline{toc}{chapter}{#1}%
	\fi
	\chaptermark{#1}%
	\addtocontents{lof}{\protect\addvspace{10\p@}}%
	\addtocontents{lot}{\protect\addvspace{10\p@}}%
	\if@twocolumn
	\@topnewpage[\@makechapterhead{#2}]%
	\else
	\@makechapterhead{#2}%
	\@afterheading
	\fi}
\def\@makechapterhead#1{%
	\vspace*{50\p@}%
	{\parindent \z@ \raggedright \normalfont
		\ifnum \c@secnumdepth >\m@ne
		\huge\bfseries \@chapapp\space \thechapter
		\par\nobreak
		\vskip 20\p@
		\fi
		\interlinepenalty\@M
		\Huge \bfseries \hyperlink{chap\thechapter}{#1}\par\nobreak
		\vskip 40\p@
	}}
	\makeatother

\newtheorem{ass}{Assumption}

\newtheorem{cor}{Corollary}[section]

\newtheorem{thrm}{Theorem}[section]

\usepackage{natbib}
\title{\Large \bf Optimal Stochastic Decensoring and Applications to Calibration of Market Models} %
\author{\textsc{Anastasis Kratsios}
		\thanks{Department of Mathematics and Statistics, Concordia University, 1455 Boulevard de Maisonneuve Ouest, Montr\'{e}al,
			Qu\'{e}bec, Canada H3G 1M8.
            Corresponding Author: email: anastasis.kratsios@mail.concordia.ca}
		\textsc{
	}\\
}
\date{
\today \\ %
}

\newcommand{\uu}{\mathbb{U}}
\newcommand{\lus}{\mathscr{L}_s^u}
	\newcommand{\lusrev}{\mathscr{L}_{1-s}^u}
\newcommand{\fty}[1]{
		\ensuremath{\mathscr{F}_t^{#1}}
	}

\begin{document}
	\maketitle
\begin{abstract}
	Typically flat filling, linear or polynomial interpolation methods to generate missing historical data.  We introduce a novel optimal method for recreating data generated by a diffusion process.  
	The results are then applied to recreate historical data for stocks.   
\end{abstract}
\textit{Keywords:} Optimal Decensoring, Stochastic Filtering, Backfilling, Time-Reversal, Conditional Stochastic Differential Equations. \hfill\\
\section{Introduction}
Many calibration method typically require the use of a greate amount of data than is available.  Often times this is delt with by bootstrapping the data using a variety of different techniques.  Here, we present and solve the optimal way to re-create historical data in a way which is consitent with the sparse prices observed in the market.   

Let $\beta_t$ be an $\rrD$-valued diffusion process, modelling a set of $D$-observable \textit{benchmarks} for which there is no missing data.  Let $\eta_t$ be a diffusion process modelling the true trajectory of the price of the asset we are interested in decensoring.  Moreover let $N_t$ be a poisson point process with intensity $\lambda$ capturing the times which $\eta_t$ is observed before time $T$, after which there are no missing data issued.  We model the observed prices of the censored process as the compound Poisson process $Y_t$ given by
\begin{equation}
\begin{aligned}
Y_t\triangleq & I_{t< T}\int_0^t \eta_{s-}dN_s + I_{t\geq T} \eta_t.  \\
\eta_t =& \eta_0 +\int_0^t h(t,\beta_t,X_t)dt + K(s) dW_t^{\eta}.
\end{aligned}
\label{eq_censorisky}
\end{equation}
where $W_t^{\eta}
$
is a Brownian motion 
and $h,K$ are deterministic functions of appropriate dimension.  

We assume that the observer is a time \begin{equation}
t\geq T_0> T.
\end{equation}
Using the benchmarks and the current time we wish to find the optimal estimate of $X_t$ and therefore of $Y_t$ and $\eta_t$ given our current observations.  

Our solution is taken in two steps.  We first solve the problem when $N_t$ does not jump; this represents the case where no obsrevations are made prior to $T_0$ such as in the case of an initial product offering.  In the second step we assume that $N_t$ has made jumps prior to time $T_0$.  Note that, since the observer's vantage point is after this time, these price times are not stochastic.  However the issue of only picking paths which match up with the realised price data $\left(T_i,\eta_{T_i}\right)$, where $T_i$ are the times prior to $T_0$ when $N_t$ jumps, presents a new issue.  This is solved by truning to conditional stochastic differential equations.  This theory allows us to minimally purturb the real-world measure until all the points are hit when backfilling the data.  

We begin by reviewing some of that relevant theory.  
\subsection{Stochastic Filtering}
A staple technique in mathematical finance for forecasting the future movements of a portfolio of assets is stochastic filtering.  The stochastic filtering problem aims to estimate the conditional density $\pi_t$ of an unobservable signal process $X_t$ given the information of an observable process $Y_t$.  Specifically, since 
$$
\operatorname{argmin}_{Z \in \fty{Y}}\ee\left[
\left(
Z-f(X_t)
\right)^2
\right]
= \ee \left[
f(X_t)|\fty{Y}
\right]
,
$$
where $\fty{Y}\triangleq \sigma\left\{Y_s:0\leq s\leq t \right\}$ is the $\sigma$-algebra generated by the observation process $Y_t$ up until time $t$ Proposition 2.3.3 \cite{van2007stochastic}.  

Formally, assume that the signal process $X_t$ is an $n$-dimension diffusion process following dynamics of the form
\begin{equation}
X_t = X_0 + \int_0^t b(s,X_s,u_s)ds + \int_0^t\sigma(s,X_s,u_s)dW_s
,
\label{eq_stoch_bakcfill_signal}
\end{equation}
where $W_t$ is an $m$-dimension Wiener process.  Moreover assume that the observation process $Y_t$ is a diffusion process following the dynamics
\begin{equation}
\eta_t = \int_0^t h(s,X_s,u_s)ds + \int_0^tK(s)dW_s^{\eta}
,
\label{eq_stoch_bakcfill_observ}
\end{equation}
where $W_t^{\eta}$ is a $p$-dimensional Brownian motion and that, $u_t$ is a control with domain $\uu$ a Borel subset of $\rr^q$, $b:[0,\infty)\times \rr^n\times \uu \rightarrow \rr^n$, $\sigma:[0,\infty)\times \rr^n\times \uu \rightarrow \rr^{n\times m}$, $h: [0,\infty)\times \rr^n\times \uu \rightarrow \rr^q$, and $K:[0,\infty)\rightarrow \rr^{p\times p}$ are measurable maps and the control $u_t$ is adapted to the filtration $\fty{Y}$.  
It was shown that the optimal estimates $\pi_t(f)\triangleq \ee \left[
f(X_t)|\fty{Y}
\right]$ according to the mean-squared error of the signal's density must follow the dynamics
\begin{equation}
\pi_t(f)=\pi_0(f) + \int_0^t \pi_s(\lus f) ds + \int_0^t \left[
K(s)^{-1}\left(
\pi_s(h_s^u(f)
-
\pi_s(f)\pi_s(h_s^u)
\right)
\right]^{\star}d\bar{B}_s
,
\end{equation}
where $\star$ denotes the transpose operator and $\bar{B}_t$ is an $\fty{Y}$-Brownian motion called the innovations process defined by
\begin{equation}
\begin{aligned}
\bar{B}_t \triangleq &
\bar{Y}_t -
\int_0^t \pi_s\left(
K(s)^{-1}h_s^u
\right)ds \\
\bar{Y}_t \triangleq & \int_0^t K(s)^{-1}dY_s
,
\end{aligned}
\label{eq_innovations_process}
\end{equation}
and where $\lus$ is the infinitesimal generator of the diffusion $X_t$ given by
\begin{equation}
\begin{aligned}
\lus f(t,x) \triangleq &
\frac{\partial f}{\partial t} (t,x)
+ 
\sum_{i=1}^n b^i(t,x,u)\frac{\partial f}{\partial x_i} (t,x) \\
+ &
\frac{1}{2} \sum_{i,j=1}^n\sum_{k=1}^m \sigma(t,x,u)^{ik}\sigma(t,x,u)^{jk} 
\frac{\partial^2 f}{\partial x_i\partial x_j} (t,x)
, 
\end{aligned}
\label{eq_infinitesimal_gen_diff}
\end{equation}
(see Proposition 7.2.9 \cite{van2007stochastic} for more details).  

\subsection{Conditional Stochastic Differential Equations}
A classical problem in probability theory is the construction of a Brownian bridge.  Here a given Brownian motion is conditioned to hit a prespecified point at a predetermined time.  This general procedure is developed in \cite{baudoin2002conditioned}, where for a fixed time-horizon $T$ an $\fff_T$-measurable $\rr$-valued random variable $Y$ conditioning the trajectories of $X_t$ and a probability measure $\nu$ describing the probability for a conditioned trajectory to happen.  The triple $(T,Y,\nu)$ is called a conditioning on the corresponding Wiener space and exists under the following assumptions.
\begin{ass}\label{ass_2}\hfill
	\begin{itemize}
		\item For $0\leq t<T$, the law of $B_t$ under $\pp$ is absolutely continuous with respect to the law of 
		$$
		\pp|_{\fff_t}\otimes \pp_Y,
		$$
		when $(\Omega,\ffft,\pp)$ is endowed with the filtration 
		$$
		\bar{\fff}_t\triangleq \fff_t \otimes B(\rr),
		$$
		where $B(\rr)$ is the Borel $\sigma$-algebra on $\rr$ and $\pp_Y$,
	\item $Supp(\nu)\subseteq Supp(\pp)$,
	\item $L^1_{\pp_Y}(\rr)\subseteq L^1_{\nu}(\rr)$.  
	\end{itemize}
\end{ass}
Under these assumptions it is shown in \citep[Proposition 2.1]{baudoin2002conditioned} that there exists a unique probability measure making the law of $Y$ into $\nu$ and respecting conditional expectation on bounded random-variables.  

Moreover, it can be seen in \citep[Proposition 2.2]{baudoin2002conditioned} that this measure minimally deviates from the reference measure $\pp$ in the sense that amongst all the absolutely continuous measure with respect to $\pp$, it has both the least entropy and its Radon-Nikodym derivative has the least variance.  
\section{Optimal Stochastic Decensoring}
The unconditional problem is first solved.  Subsequently we address the full-decensoring problem.  
%
\subsection{Optimal Backfilling}
In principal we should have enough data to calibrate $\pi_t(f)$ on some time interval $[1,1+\epsilon]$ for some $\epsilon>0$.  Once we have obtained the dynamics of $\pi_t(f)$ we can then reverse time and backfill the missing data using the dynamics of $\pi_{1-t}(f)$ described by the following equations.  
\begin{ass}[Filtering Assumptions]\label{ass_reg_con_backwardstime}\hfill
\begin{itemize}
\item There exists a constant $K$ such that for every $x,y \in \rr^n$
\begin{equation}
\begin{aligned}
\operatorname{sup}_{t,u} &
\left[
\left|
\sigma(t,x,u)-\sigma(t,y,u)
\right|
+ \left|
b(t,x,u)-b(t,y,u)
\right|
\right]\leq K|x-y|\\
\operatorname{sup}_{t,u} &
\left[
\left|
\sigma(t,x,u)-\sigma(t,y,u)
\right|
\right]\leq K(1+|x|)
\end{aligned}
\end{equation}
\item For any $t >0$, $X_t$ has a density $p_t$.  
\item For every $i$ and any bounded open subset $U$ of $\rr^n$ and any $t_0>0$
$$
\int_{t_0}^1 \int_{U} \left|
\nabla_j \left[
	\sigma^i(t,x,u)\sigma^j(t,x,u)p_t(x)
\right]
\right|dxdt <\infty
.
$$
\item $(X_t,Y_t)$ has a unique $\fty{Y}$-adapted solution.
\item $K(t)$ is invertible for every $t$.
\item $K(t)$ and $K(t)^{-1}$ are locally bounded.
\item $b,\sigma,h$ are bounded functions.  
\end{itemize}
\end{ass}
\begin{ass}[Time-Reversal Assumptions]\label{ass_4_TR}
	The sums of the distributional derivatives $$
	\sum_{i=1}^{d}\left(\sigma(t,x)\sigma(t,x)^{\star}\right)P(t,x),
	$$
	are locally integrable.  
\end{ass}
The backfilling problem can be formalized as seeking the dynamics for the time-reversed conditional expectation, given the current information.  That is, we seek dynamics for the optional projection of the process $b_t(f)$ defined by
$$
b_{1-t}(f)\triangleq \ee\left[
f(X_t)\mid \fff_t^{Y}
\right] = \operatorname{argmin}_{Z \in L^2(\fff_t^Y)}\ee\left[
\left(
Z - f(X_t)
\right)^2
\right].  
$$
The solution to this problem is now presented.  
\begin{thrm}[Optimal Backfilling Equations]\label{thrm_stoch_backfilling_equations}
	Assume that $X_t$ and $Y_t$ are as described in equations \eqref{eq_stoch_bakcfill_signal} and \eqref{eq_stoch_bakcfill_observ} respectively and that the assumptions \autoref{ass_reg_con_backwardstime} and \autoref{ass_4_TR} hold.  Then the dynamics of $\pi_{1-t}(f)$ for any $f \in C_b^{1,2}(\rr,\rr^n)$ are given by
	\begin{equation}
	\begin{aligned}
	db_{t}(f)=d\pi_{1-t}(f) = & \int_{t_0}^t \left[
	K(1-s)^{-1}\left(
	\pi_{1-s}(h_{1-s}^u(f))-\pi_{1-s}(f)\pi_{1-s}(h_{1-s}^u)
	\right)
	\right]^{\star}dB_s
	\\
	+\int_{t_0}^t &\left[
	-\pi_{1-s}\left(
	\lusrev f
	\right)dt
	+ 
	\frac{
		\nabla \left[
		K(s)^{-1}\left(
		\pi_s(h_s^u(f)
		-
		\pi_s(f)\pi_s(h_s^u)
		\right)
		\right]^{\star}
	}{
		p_t(x)
	}
	\right]dt
	\end{aligned}
	\label{eq_backfilling_dynamics_1}
	,
	\end{equation}
	where $p_t$ is the density of $\pi_t(f)$ and satisfies the SPDE
	\begin{equation}
	\begin{aligned}
	p_t = & p_0(x)+ \int_0^t \left(\mathscr{L}_t^u\right)^{\star}p_t(x)dt \\
	+ &
	p_t(x)\left[
	K(t)^{-1}\left(
	h(t,x,u_t) - \pi_t(h_t^u)
	\right)
	\right]^{\star}
	\left[
	d\bar{Y}_t - \pi_t\left(
	K(t)^{-1}h_t^u
	\right)
	\right]
	.
	\end{aligned}
	\end{equation}  
\end{thrm}
\begin{proof}
	Proposition \cite[7.2.9]{van2007stochastic} implies that the dynamics for $\pi_t(f)$ are given by
	\begin{equation}
	\pi_t(f)=\pi_0(f) + \int_0^t \pi_s(\lus f) ds + \int_0^t \left[
	K(s)^{-1}\left(
	\pi_s(h_s^u(f)
	-
	\pi_s(f)\pi_s(h_s^u)
	\right)
	\right]^{\star}d\bar{B}_s
	\label{eq_restaete_pidynamics}
	.
	\end{equation}
	Under the filtering assumptions \autoref{ass_reg_con_backwardstime}, Theorem \cite[2.3]{millet1989integration}
	$\pi_{1-t}(f)$ is a diffusion process whose infinitesimal generator $\bar{L}_t^u$ is given by
	\begin{equation}
	\begin{aligned}
	\bar{L}_t^u = &\left[
	K(1-s)^{-1}\left(
	\pi_{1-s}(h_{1-s}^u(f))-\pi_{1-s}(f)\pi_{1-s}(h_{1-s}^u)
	\right)
	\right]^{\star}\Delta
	\\
	+ &\left[
	-\pi_{1-s}\left(
	\lusrev f
	\right)dt
	+ 
	\frac{
		\nabla \left[
		K(s)^{-1}\left(
		\pi_s(h_s^u(f)
		-
		\pi_s(f)\pi_s(h_s^u)
		\right)
		\right]^{\star}
	}{
		p_t(x)
	}
	\right]\nabla
	\end{aligned}
	\label{eq_infinitesimal_gen_reverse_timeio}.  
	\end{equation}
	By theorem 2.1 \cite{millet1989integration} the process $\pi_t(f)$ is Malliavin differentiable and so we may integrate $\pi_t(f)= \int f(x)p_t(x)dx$ by parts to obtain the dynamics for $dp_t(x)$ by
	\begin{equation}
	\begin{aligned}
	p_t(x) = & p_0(x)+ \int_0^t \left(\mathscr{L}_t^u\right)^{\star}p_t(x)dt \\
	+ &
	p_t(x)\left[
	K(t)^{-1}\left(
	h(t,x,u_t) - \pi_t(h_t^u)
	\right)
	\right]^{\star}
	\left[
	d\bar{Y}_t - \pi_t\left(
	K(t)^{-1}h_t^u
	\right)
	\right]dt
	.
	\end{aligned}
	\end{equation}
	Lastly, note that time begins from $t_0$ and not $0$.  Making the linear time-change $t\mapsto t+t_0$, gives the result.  
\end{proof}
In practice the added assumption of linear dynamics 
\begin{equation}
\begin{aligned}
dX_y = & A(t)X_tdt + D(t)u_tdt +C(t)dW_t \\
dY_t = & H(t)X_tdt + K(t)dB_t
\end{aligned}
\label{eq_Klamnalinear_dynamics},
\end{equation}
where $A,C,D,H$ and $K$ are matrix-valued functions of dimension $n\times n$, $n\times m$, $n\times k$, $p \times n$ and $p \times p$, respectively and $u_t$ is a $k$-dimensional control adapted to $\fty{Y}$, is common due to its computational speed of such models and simplicity of implementation.  Under these additional assumption we may simplify the dynamics of $\pi_t(f)$ of theorem 
.
\begin{thrm}[Kalman-Bucy Backfilling]\label{cor_KB_Backfilling}
	Assume that $X_0$ is Gaussian with mean $\hat{X}_0$ and covariance $\hat{P}_0$, the solution to equation \eqref{eq_Klamnalinear_dynamics} is unique, $K(t)$ is invertible for every $t$ and \hfill\\
	$A(t),C(t),D(t),H(t),K(t),K(t)^{-1}$ are continuous.  Then the dynamics of the time-reversed conditional mean 
	$$
	\hat{X}_t \triangleq \ee\left[
	X_t | \fty{Y}
	\right] 
	$$ 
	is given by
	\begin{equation}
	\begin{aligned}
	d\hat{X}_{1-t} = & \left[
	-\left(
	A(t)\hat{X}_t + D(t)u_t
	\right)
	+ \frac{1}{\phi(\hat{X}_t,\hat{P}_t)}
	\left(
	\nabla(
	\hat{P}_t(K_t^{-1}H_t)
	)^{\star}\phi(\hat{X}_t,\hat{P}_t)
	\right)
	\right]dt\\
	+ &\hat{P}_t(K(t)^{-1}H(t))^{\star}d\bar{B}_t \\
	\frac{d\hat{P}_t}{dt} = &
	A(t)\hat{P}_t + \hat{P}_tA(t)^{\star} - \hat{P}_tH(t)^{\star}(K(t)K(t))^{\star}
	H(t)\hat{P}_t + C(t)C(t)^{\star}\\
	d\bar{B}_t = & K(t)^{-1}(dY_t - H(t)\hat{X}_t dt),
	\end{aligned}
	\label{eq_conditional_mean}
	\end{equation}
	$\bar{B}_t$ is a $\fty{Y}$-Brownian motion, and $\phi(\mu,\Sigma)$ is the density of a Gaussian random variable with mean $\mu$ and covariance $\Sigma$.  
\end{thrm}
\begin{proof}
	If $X_t$ and $Y_t$ follow the linear dynamics of equation \eqref{eq_Klamnalinear_dynamics}, then the assumptions of theorem 
	imply that \cite[theorem 7.3.1]{van2007stochastic} the dynamics of $X_t$ are given by
	\begin{equation}
	\begin{aligned}
	d\hat{X}_{t} = & A(t)\hat{X}_{t} dt + D(t)u_tdt + \hat{P}_{t}(K(t)^{-1}H(t))^{\star} d\bar{B}_t \\
	\frac{d\hat{P}_t}{dt} = &
	A(t)\hat{P}_t + \hat{P}_tA(t)^{\star} - \hat{P}_tH(t)^{\star}(K(t)K(t))^{\star}
	H(t)\hat{P}_t + C(t)C(t)^{\star}\\
	d\bar{B}_t = & K(t)^{-1}(dY_t - H(t)\hat{X}_t dt)
	.
	\end{aligned}
	\label{eq_cor_mamammaman}
	\end{equation}
	Since $\hat{X}_t$ follows linear dynamics and $A,B,C,H,K$ are deterministic then $X_t$ follows a normal law and therefore has a well-defined density $\phi$.  Therefore theorem 2.3 of \cite{millet1989integration} applies to the dynamics of $X_t$ which yields equation 
\end{proof}
\subsection{Optimal Stochastic Decensoring}
Assume that $N_t$ is possibly greater than $0$ before time $T_0$.  
The optimal stochastic decensoring problem can be formalized as seeking the dynamics for the time-reversed conditional expectation, given the current information.  That is, we seek dynamics for the optional projection of the process $\Xi_t(f)$ defined by
\begin{equation}
\begin{aligned}
b_{1-t}(f)\triangleq& \ee_{\pp}\left[
f(X_t)\mid \fff_t^{Y}
\right] = \operatorname{argmin}_{Z \in L^2(\fff_t^Y)}\ee_{\pp}\left[
\left(
Z - f(X_t)
\right)^2
\right]\\
\Xi_t \triangleq & \frac{d\qq}{d\pp} b_t\\
\qq \triangleq & \operatorname{arginf}\left\{H\left(Q\|\pp\right) \mid
Q\ll \pp \mbox{ and } Q\left(\min\left[I\left(\xi_{T_i}=\eta_{T_i}\right]\right)=1\right)=1.  
\right\}
\end{aligned}
\label{eq_definitions_Hyper_f_1}
\end{equation}
We interpret $\Xi_t$ as the stochastic backfiller $b_t$ but under a point of view on the world given by $\qq$ conditioned an the observed data $\left(T_i,\eta_{T_i}\right)_{i=1}^N$, indeed being realized.  Therefore, this perturbation of the stochastic backfiller assures that the observations become realizations of when filling backwards in time according to $\Xi_t$.  
\begin{ass}[Hormander Conditions]\label{ass_3}\hfill
	\begin{itemize}
\item For every $x\in \rr^d$, the SDE of Equation \eqref{eq_backfilling_dynamics_1} has a unique strong solution for $b_t(f)=x$,
\item The semi-group $P_t(x,dt)$ of the SDE of Equation \eqref{eq_backfilling_dynamics_1} has smooth densities.  
	\end{itemize}
\end{ass}
\begin{thrm}[Optimal Stochastic Decensoring]\label{thrm_Optimal_Stochastic_Dencensoring_main_1}
	Assume that $X_t$ and $Y_t$ are as described in equations \eqref{eq_stoch_bakcfill_signal} and \eqref{eq_stoch_bakcfill_observ} respectively and that the the assumptions \autoref{ass_reg_con_backwardstime} and \autoref{ass_4_TR} hold.  Moreover, make assumptions \autoref{ass_2} as well as the Hormander conditions \autoref{ass_3}.  	
	
	Then there exists a unique measure $\qq\ll\pp$ satisfying \eqref{eq_definitions_Hyper_f_1}, such that the CSDE associated to the conditioning $(T_1,Z_1,\delta)$ the dynamics of $\pi_{1-t}(f)$ for any $f \in C_b^{1,2}(\rr,\rr^n)$ are given by
	\begin{equation}
	\begin{aligned}
	db_{t}(f)=d\pi_{1-t}(f) = & \int_{t_0}^t \left[
	K(1-s)^{-1}\left(
	\pi_{1-s}(h_{1-s}^u(f))-\pi_{1-s}(f)\pi_{1-s}(h_{1-s}^u)
	\right)
	\right]^{\star}dB_s
	\\
	\int_{t_0}^t &\phi(t,p_t)+
	\left[
	-\pi_{1-s}\left(
	\lusrev f
	\right)dt
	+ 
	\frac{
		\nabla \left[
		K(s)^{-1}\left(
		\pi_s(h_s^u(f)
		-
		\pi_s(f)\pi_s(h_s^u)
		\right)
		\right]^{\star}
	}{
		p_t(x)
	}
	\right]dt
	\end{aligned}
	\label{eq_backfilling_dynamics}
	,
	\end{equation}
	where $p_t$ is the density of $\pi_t(f)$ and $\phi$ and $p_t$ satisfy
	\begin{equation}
	\begin{aligned}
	p_t = & p_0(x)+ \int_{t_0}^t \left(\mathscr{L}_t^u\right)^{\star}p_t(x)dt \\
	+ &
	p_t(x)\left[
	K(t)^{-1}\left(
	h(t,x,u_t) - \pi_t(h_t^u)
	\right)
	\right]^{\star}
	\left[
	d\bar{Y}_t - \pi_t\left(
	K(t)^{-1}h_t^u
	\right)
	\right]\\
	0 = &\frac{\partial \phi}{\partial t}
	+ 
	\frac{\partial}{\partial x}\left(
			\left[
			-\pi_{1-s}\left(
			\lusrev f
			\right)dt
			+ 
			\frac{
				\nabla \left[
				K(s)^{-1}\left(
				\pi_s(h_s^u(f)
				-
				\pi_s(f)\pi_s(h_s^u)
				\right)
				\right]^{\star}
			}{
				p_t(x)
			}
			\right]
	\phi
	\right)\\
	+&
	\frac{1}{2}
		\nabla
	\left(
	\left(
		\left[
		K(1-s)^{-1}\left(
		\pi_{1-s}(h_{1-s}^u(f))-\pi_{1-s}(f)\pi_{1-s}(h_{1-s}^u)
		\right)
		\right]^{\star}
	\right)^2 \phi^2
	\right)
		\\
	+&
	\frac{1}{2}\nabla
	\left(
			\left(
	\left[
	K(1-s)^{-1}\left(
	\pi_{1-s}(h_{1-s}^u(f))-\pi_{1-s}(f)\pi_{1-s}(h_{1-s}^u)
	\right)
	\right]^{\star}
		\right)^2
	\nabla\phi
	\right)
	.
	\end{aligned}
	\end{equation}  
\end{thrm}
\begin{proof}
Consider the dynamics for $b_t$ obtained in theorem \autoref{thrm_stoch_backfilling_equations} under the conditioning $(T_0-T_i,Y_{T_i},\nu)$ where $\nu$ is the standard Gaussian measure on $\rr^d$.  Proposition \citep[Proposition 3.1]{baudoin2002conditioned} gives the result.  
\end{proof}

Theorem \autoref{thrm_Optimal_Stochastic_Dencensoring_main_1} may be difficult to apply directly in practice.  However, the points $(Y_{T_i})_{i=1}^N$ are $\fff_{T_0}$-measurable.  Therefore, we instead interpolating the points $(Y_{T_i})_{i=1}^N$ by the following process.  As an alternative we may turn to a relaxed, interpolation-type approach.  

\begin{cor}
	The process $b_t$ defined by
\begin{equation*}
\begin{aligned}
\hat{b}_t(f) = b_t(f) + \sum_{i=1}^{K}\left[
\frac{
	\left(Y_{\tau_{i+1}}-b_{\tau_{i+1}}(f)\right)\left(\tau_{i+1}-t\right)
+ 
	\left(Y_{\tau_{i+1}}-b_{\tau_{i}}(f)\right)\left(t-\tau_{i}\right)
}{
	\tau_{i+1}-\tau_i
}
I\left(\tau_i\leq  t \leq \tau_{i+1}\right)
\right]
\end{aligned}
\end{equation*}  
\begin{enumerate}[{i}]
\item For $-t\geq T_0$, $\hat{t}=\pi_t(f)$,
\item For every $i$ in $1,\dots,K$, $\hat{b}_{\tau_i}=Y_{T_i}$,
\item Has $\pp$-a.s. continuous paths.  
\end{enumerate}
In particular if there exists a measure $Q\ll\pp$ such that $b_t(f)\frac{dQ}{d\pp} = \hat{b}_t$, then
\begin{equation}
H\left(Q|\pp\right)\geq H\left(\qq|\pp\right)
.
\label{eq_relaxation}
\end{equation}
\end{cor}
\begin{proof}
Properties $i-iii$ hold by construction of $\hat{b}_t$.  Suppose that a measure $Q\ll\pp$ satisfying $b_t(f)\frac{dQ}{d\pp} = \hat{b}_t$ exists.  Then equation \eqref{eq_relaxation} holds by \citep[Proposition 2.1]{baudoin2002conditioned}.  
\end{proof}

%
%
\section{Conclusion}
We have developed an extension of the stochastic filtering framework which is able to capture and recreate missing historical data in an optimal manner, given the presently available market data.  This method works by benchmarking an asset against a set of liquid assets during a time period when it is itself liquid, estimating the forward-filter itself and then recreating the missing data by back-propagating the filter.  

The backwards filter is perturbed by choosing the optimal measure assuring that the realized data points before the point of liquidity are attained in path simulation process.

\bibliographystyle{abbrvnat}
\bibliography{References_1}
\printindex
\end{document}